\relax
\documentclass[12pt]{article}
\usepackage{natbib}
\usepackage{times}
\usepackage{helvet}
\usepackage{courier}
\usepackage{enumerate,setspace}
\usepackage{amssymb,latexsym}
\usepackage{graphicx,amsmath,amsfonts,amsthm}
\usepackage{comment}
\usepackage{tikz}
\usepackage{tkz-graph}
\usepackage{pgfplots}
\usepgfplotslibrary{fillbetween}
\usepackage{url}
\usepackage{mathrsfs}
\usepackage{paralist}
\usepackage[]{units}
\usepackage{microtype}
\usepackage{booktabs}
\usepackage[english, vlined, ruled, linesnumbered]{algorithm2e}
\usepackage[small]{caption}
\usepackage{subcaption}
\usepackage{etoolbox}
\usepackage{hyperref}
\usepackage[sort&compress,nameinlink]{cleveref}

\usepackage{etoolbox}

\hypersetup{
colorlinks=true,citecolor=blue,linkcolor=red!60!black,
pagebackref=true
}

\newcommand{\DCF}{{\sc DCF}}
\newcommand{\DCWD}{{\sc DCWD}}
\newcommand{\DCFlong}{{\sc Diverse Committee Feasibility}}
\newcommand{\DCWDlong}{{\sc Diverse Committee Winner Determination}}
\newcommand{\BCWD}{{\sc BCWD}}

\newcommand{\ch}{\mbox{\small \rm child}}

\newcommand{\tblOpt}{\mathsf{Opt}}
\newcommand{\nbchld}{{\rm \#children}}
\newcommand{\chld}{{\rm child}}
\newcommand{\desc}{{\rm desc}}
\newcommand{\best}{{\rm best}}

\newcommand{\vertllb}{{\rm low}}
\newcommand{\verthlb}{{\rm high}}
\newcommand{\edgelb}{{\rm ed}}
\newcommand{\id}{{\rm id}}

\newcommand{\reals}{\mathbb R}

\newcommand{\naturals}{\mathbb N}

\newcommand{\calK}{\mathcal{K}}

\newcommand{\calI}{\mathcal{I}}

\newcommand{\bfw}{\boldsymbol{w}}

\newcommand{\np}{{{\mathrm{NP}}}}
\newcommand{\fpt}{{{\mathrm{FPT}}}}

\newcommand{\pos}{{{{\mathrm{pos}}}}}

\def\argmax{\mbox{argmax}}
\theoremstyle{plain}
	  \newtheorem{theorem}{Theorem}
	  
	  \newtheorem{lemma}[theorem]{Lemma}
	  \newtheorem{prop}[theorem]{Proposition}
	  \newtheorem{claim}[theorem]{Claim}
\theoremstyle{definition}
	  \newtheorem{definition}[theorem]{Definition}
	  
	  \newtheorem{example}{Example}
	  
\theoremstyle{remark}

\makeatletter

\patchcmd{\NAT@citex}
  {\@citea\NAT@hyper@{%
     \NAT@nmfmt{\NAT@nm}%
     \hyper@natlinkbreak{\NAT@aysep\NAT@spacechar}{\@citeb\@extra@b@citeb}%
     \NAT@date}}
  {\@citea\NAT@nmfmt{\NAT@nm}%
   \NAT@aysep\NAT@spacechar\NAT@hyper@{\NAT@date}}{}{}

\patchcmd{\NAT@citex}
  {\@citea\NAT@hyper@{%
     \NAT@nmfmt{\NAT@nm}%
     \hyper@natlinkbreak{\NAT@spacechar\NAT@@open\if*#1*\else#1\NAT@spacechar\fi}%
       {\@citeb\@extra@b@citeb}%
     \NAT@date}}
  {\@citea\NAT@nmfmt{\NAT@nm}%
   \NAT@spacechar\NAT@@open\if*#1*\else#1\NAT@spacechar\fi\NAT@hyper@{\NAT@date}}
  {}{}

\makeatother

\title{Multiwinner Elections with Diversity Constraints
}
\author{
Robert Bredereck \\ University of Oxford, Oxford, UK;\\ TU Berlin, Berlin, Germany \\ robert.bredereck@tu-berlin.de \and
Piotr Faliszewski \\ AGH University, Krakow, Poland \\ faliszew@agh.edu.pl \and
Ayumi Igarashi \\ University of Oxford, Oxford, UK \\ ayumi.igarashi@cs.ox.ac.uk \and
Martin  Lackner  \\ TU Wien, Vienna, Austria \\ lackner@dbai.tuwien.ac.at  \and
Piotr  Skowron \\ TU Berlin, Berlin, Germany \\ p.k.skowron@gmail.com
}
\date{}

\begin{document}

\maketitle

\begin{abstract}
  % Many multiwinner elections have diversity requirements that are
  % independent from judging the individual candidates (consider
  % shortlisting job applicants and maintaining representation of
  % certain minorities).  
  We develop a model of multiwinner elections that combines
  performance-based measures of the quality of the committee (such as,
  e.g., Borda scores of the committee members) with diversity
  constraints. Specifically, we assume that the candidates have
  certain attributes (such as being a male or a female, being junior
  or senior, etc.) and the goal is to elect a committee that, on the
  one hand, has as high a score regarding a given performance measure,
  but that, on the other hand, meets certain requirements (e.g., of
  the form ``at least $30\%$ of the committee members are junior
  candidates and at least $40\%$ are females'').  We analyze the
  computational complexity of computing winning committees in this
  model, obtaining polynomial-time algorithms (exact and approximate)
  and $\np$-hardness results. We focus on several natural classes of
  voting rules and diversity constraints.
% 
%
%
  % Many applications of committee elections require explicit diversity
  % requirements which cannot be handled by standard models from the
  % literature.  To fill this gap, we develop a framework where
  % candidates have (possibly multiple) labels (e.g.\ ``female/male'' or
  % ``senior/junior/student'') and the committee to be elected has to
  % fulfill certain diversity constraints (e.g.\ ``out of ten committee
  % members there must be at least three students'').  We analyze the
  % computational complexity of finding diverse committees with respect
  % to several (classes of) voting rules (e.g.\ submodular, separable,
  % and Chamberlin-Courant).  In this context, we define and take into
  % account natural restrictions of the labeling structure and of the
  % diversity constraints.  We also consider the special case of
  % balanced diversity constraints (requiring the same number of
  % candidates for each of two labels) and develop a polynomial-time
  % algorithm that computes a balanced committee which is approximately
  % optimal with respect to Chamberlin-Courant.  Finally, we define the
  % price of diversity as ratio between the optimal score of arbitrary
  % committees and the optimal score of diverse committees and show that
  % this is upper-bounded by two for balanced diversity constraints.
\end{abstract}

\section{Introduction}
We study the problem of computing committees (i.e., sets of
candidates) that, on the one hand, are of high quality (e.g., consist
of high-performing individuals) and that, on the other hand, are
diverse (as specified by a set of constraints).
The following example shows our problem in more concrete terms.

Consider an organization that wants to hold a research meeting on some
interdisciplinary topic such as, e.g., ``AI and Economics.'' The
meeting will take place in some secluded location and only a certain
limited number of researchers can attend. How should the organizers
choose the researchers to invite? If their main criterion were the
number of highly influential AI/economics papers that each person
published, then they would likely end up with a very homogeneous group
of highly-respected AI professors. Thus, while this criterion
definitely should be important, the organizers might put forward
additional constraints. For example, they could require that at least
30\% of the attendees are junior researchers, at least 40\% are
female, at least a few economists are invited (but only senior ones),
the majority of attendees work on AI, and the attendees come from at
least 3 continents and represent at least 10 different
countries.\footnote{For example, the Leibniz-Zentrum f\"ur Informatik
  that runs Dagstuhl Seminars gives similar suggestions to event
  organizers.}  In other words, the organizers would still seek
researchers with high numbers of strong publications, but they would
give priority to making the seminar more diverse (indeed, junior
researchers or representatives of different subareas of AI can provide
new perspectives; it is also important to understand what people
working in economics have to say, but the organizers would prefer to
learn from established researchers and not from junior ones).

The above example shows a number of key features of our
committee-selection model.
First, we assume that there is some function that evaluates the
committees (we refer to it as the \emph{objective function}). In the
example it was (implicitly) the number of high-quality papers that the
members of the committee published. In other settings (e.g., if we
were shortlisting job candidates) these could be aggregated opinions
of a group of voters (the recruitment committee, in the shortlisting
example).

Second, we assume that each prospective committee member (i.e., each
researcher in our example) has a number of attributes, which we call
labels. For example, a researcher can be \emph{junior} or
\emph{senior}, a \emph{male} or a \emph{female}, can \emph{work in AI}
or in \emph{economics} or in some other area, etc.  Further, the way
in which labels are assigned to the candidates may have a structure on
its own. For example, each researcher is either male or female and
either junior or senior, but otherwise these attributes are
independent (i.e., any combination of gender and seniority level is
possible). Other labels may be interdependent and may form
hierarchical structures (e.g., every researcher based in Germany is
also labeled as representing Europe).  Yet other labels may be
completely unstructured; e.g., researchers can specialize in many
subareas of AI, irrespective how (un)related they seem.

Third, we assume that there is a formalism that specifies when a
committee is \emph{diverse}. In principle, this formalism could be
any function that takes a committee and gives an
\emph{accept/reject} answer. However, in many typical settings it
suffices to consider simple constraints that regard each label
separately (e.g., ``at least 30\% of the researchers are junior'' or
``the number of male researchers is even''). We focus on such
independent constraints, but studying more involved ones, that regard
multiple labels (e.g., ``all invited economists must be senior
researchers'') would also be interesting.

Our goal is to find a committee of a given size $k$ that is diverse
and has the highest possible score from the objective function.
While similar problems have already been considered (see the Related
Work section), we believe that our paper is the first to
systematically study the problem of selecting a diverse committee,
where diversity is evaluated with respect to candidate attributes.  We
provide the following main contributions:
\begin{enumerate}
\item We formally define the general problem of selecting a diverse
  committee and we provide its natural restrictions. Specifically, we
  focus on the case of submodular objective functions (with the
  special case of separable functions), candidate labels that are
  either layered or laminar,\footnote{If we restricted our example to
    labels regarding gender and seniority level, we would have
    2-layered labels (because there are two sets of labels,
    $\{\mathit{male, female}\}$ and $\{\mathit{junior, senior}\}$, and
    each candidate has one label from each set. On the other hand,
    hierarchical labels, such as those regarding countries and
    continents, are 1-laminar (see description of the model for more
    details).} and constraints that specify sets of acceptable
  cardinalities for each label independently (with the special case of
  specifying intervals of acceptable values).

\item We study the complexity of finding a diverse committee of a
  given size, depending on the type of the objective function, the
  type of the label structure, and the type of diversity constraints.
  While in most cases we find our problems to be $\np$-hard (even if
  we only want to check if a committee meeting diversity constraints
  exists; without optimizing the objective function), we also find
  practically relevant cases with polynomial-time algorithms
  (e.g., our algorithms would suffice for the research-meeting example
  restricted to the constraints regarding the seniority level and
  gender).
  We provide approximation algorithms for some of our $\np$-hard problems.

\item We study the complexity of recognizing various types of label
  structures. For example, given a set of labeled candidates, we ask
  if their labels have laminar or layered structure. It turns out that
  recognizing structures with three independent sets of labels is
  $\np$-hard, whereas recognizing up to two independent sets is
  polynomial-time computable.
  
\item Finally, we introduce the concept of \emph{price of diversity},
  which quantifies the ``cost'' of introducing diversity constraints
  subject to the assumed objective function.
  
\end{enumerate}

Our main results are presented in \Cref{table}.

%\begin{description}
%\item[Interval constraints.] A team of ten software engineers has to be assembled that is tasked with creating a smart phone application. The team must include at least two but not more than five user interface specialists. 
%\item[Independent constraints.] A rescue team consists of pilots, paramedics, and doctors. As pilots always fly in pairs, the number of pilots has to be even. 
%\item[Equality constraints.] A controversial political issue is to be decided by a committee consisting of twenty experts, activists, and politicians. Since activists oppose the solution proposed by the (responsible) politicians, it is agreed that the number of activists in the committee equals the number of politicians. The number of experts is not constrained.
%\item[Interval+equality constraints.] The situation is as in the previous example, but the committee has to include at least five experts and two activists.
%\end{description}

\section{The Model}

For $i,j\in\naturals$, we write $[i,j]$ to denote the set
$\{i,i+1,\dots,j\}$.  We write $[i]$ as an abbreviation for $[1,i]$.
For a set $X$, we write $2^X$ to denote the family of all of its
subsets.  We first present our model in full generality and then
describe the particular instantiations that we focus on in our
analysis.

\paragraph{General Model}

Let $C = \{c_1, \ldots, c_m\}$ be a set of candidates and let $L$ be a
set of labels (such as \emph{junior}, \emph{senior}, etc.).  Each
candidate is associated with a subset of these labels through a
labeling function $\lambda \colon C \rightarrow 2^L$.  We say that a
candidate~$c$ has label~$\ell$ if~$\ell \in \lambda(c)$, and we write
$C_\ell$ to denote the set of all candidates that have label $\ell$.

A diversity specification is a function that given a committee (i.e.,
a set of candidates), the set of labels, and the labeling function
provides a \emph{yes/no} answer specifying if the committee is
\emph{diverse}.  If a committee is diverse with respect to diversity
specification $D$, then we say that it is \emph{$D$-diverse}.

% In principle, we put no restrictions
% on the diversity specifications, but later we will focus on fairly
% simple conditions that regard the numbers of candidates with
% particular labels.

An objective function $f \colon 2^C \rightarrow \reals$ is a function
that associates each committee with a score. We assume that
$f(\emptyset) = 0$ and that the function is monotone (i.e., for each
two committees $A$ and $B$ such that $A \subseteq B$, it holds that
$f(A) \leq f(B)$). In other words, an empty committee has no value and
extending a committee cannot hurt it.

Our goal is to find a committee of a given size $k$ that meets the
diversity specification and that has the highest possible score according
to the objective function.

\begin{definition}[\DCWDlong{} (\DCWD{})]
  Given a set of candidates $C$, a set of labels $L$, a labeling
  function $\lambda$, a diversity specification $D$, a desired
  committee size $k$, and an objective function $f$, find a committee
  $W \subseteq C$ with $|W|=k$ that achieves the maximum value $f(W)$
  among all $D$-diverse size-$k$ committees.
\end{definition}

The set of candidates, the set of labels, and the labeling function
are specified explicitly (i.e., by listing all the candidates with all
their labels). The encoding of the diversity specification and the
objective function depends on a particular case (see discussions
below). To consider the problem's $\np$-hardness, we take its
decision variant, where instead of asking for a $D$-diverse committee
with the highest possible value of the objective function we ask if
there exists a $D$-diverse committee with objective value at least $T$
(where the threshold $T$ is a part of the input).

We also consider the \DCFlong{} (\DCF{})
problem, which takes the same input as the winner determination
problem, but where we ask if any $D$-diverse committee of size $k$
exists, irrespective of its objective value. In other words, the
feasibility problem is a special case of the decision variant of the
winner determination problem, where we ask about a $D$-diverse
committee with objective value greater or equal to $0$.  Thus, if the
feasibility problem is $\np$-hard, then the analogous winner
determination problem is $\np$-hard as well (and if the winner
determination problem is polynomial-time computable, so is the
feasibility problem).

The model, as specified above, is far to general to obtain any sort of
meaningful computational results. Below we specify its restrictions
that we study.

\paragraph{Objective Functions}
\label{subsec:satisfaction-f}

An objective
function is submodular if for each two committees $S$ and $S'$ such
that $S \subseteq S' \subseteq C$ and each $c \in C \setminus S'$ it
holds that $f(S \cup \{c\}) - f(S) \geq f(S' \cup \{c\}) -
f(S')$. 
For two sets of candidates $X$ and $S$, we write $f(X|S)$ to denote
the marginal contribution of the candidates from $X$ with respect to
those in $S$. Formally, we have $f(X|S)=f(S\cup X)-f(S)$.  
Submodular functions are very common and suffice to express many
natural problems. We assume all our objective functions to be
submodular.

\begin{example}\label{ex:cc}
  Consider the following voting scenario.  Let
  $C= \{c_1, \ldots, c_m\}$ be  a set of candidates and $V = \{v_1, \ldots,
  v_n\}$ a set of voters, where each voter ranks all the candidates from best to
  worst. We write $\pos_{v_i}(c)$ to denote the position of candidate
  $c$ in the ranking of voter $v_i$ (the best candidate is ranked on
  position $1$, the next one on position $2$, and so on).
  The Borda score associated with position $i$ (among $m$ possible
  ones) is  $\beta_m(i) = m-i$.  Under the
  Chamberlin--Courant rule (CC), the score of a committee $S$ is
  defined by objective function $ f^{\mathrm{CC}}(S) = \sum_{i=1}^n
  \beta_m(\min\{\pos_{v_i}(c) \mid c \in S\}).  $ Intuitively, this
  function associates each voter with her representative (the member
  of the committee that the voter ranks highest) and defines the score
  of the committee as the sum of the Borda scores of the voters'
  representatives. It is well-known that this function is
  submodular~\citep{budgetSocialChoice}.
  The CC rule outputs those committees (of a given size $k$) for which
  the CC objective function gives the highest value (and, intuitively,
  where each voter is represented by a committee member that the voter
  ranks highly).
\end{example}

As a special case of submodular functions, we also consider
\emph{separable} functions. A function is separable if for every
candidate $c \in C$ there is a weight $w_c$ such that the value of a
committee $S$ is given as $f(S) = \sum_{c \in S}w_c$. While separable
functions are very restrictive, they are also very natural.

\begin{example}\label{ex:k-borda}
  Consider the setting from \Cref{ex:cc}, but with objective
  function $f^{\mathrm{kB}(W)} = \sum_{i=1}^n \left( \sum_{c \in
        W}\beta_m(\pos_{v_i}(c) \right)$. This function sums Borda
    scores of all the committee members from all the voters and models
    the $k$-Borda voting rule (the committee with the highest score is
    selected).  The function is separable as for each candidate $c$ it
    suffices to take $w_c = f^{\mathrm{k\hbox{-}Borda}}(\{c\})$.  It
    is often argued that $k$-Borda is a good rule when our goal is to
    shortlist a set of individually excellent
    candidates~\citep{FSST-trends}.
\end{example}

Together, \Cref{ex:cc} and \Cref{ex:k-borda} show that our
model suffices to capture many well-known multiwinner voting
scenarios. Many other voting rules, such as Proportional Approval
Voting, or many committee scoring rules, can be
expressed through submodular objective
functions~\citep{sko-fal-lan:c:collective,fal-sko-sli-tal:c:classification}.

\paragraph{Diversity Specifications}
We focus on diversity specifications that regard each label
independently.  In other words, the answer to the question if a given
committee $S$ is diverse or not depends only on the cardinalities of the
sets $C_\ell \cap S$.

\begin{definition}%[Independent constraints] 
  For a set of candidates $C$, a set of labels $L$, and a labeling
  function $\lambda$, we say that a diversity specification $D$ is
  independent (consists of independent constraints) if and only if
  there is a function $b \colon L \rightarrow 2^{[|C|]}$ (referred to
  as the cardinality constraint function) such that a committee $S$ is
  diverse exactly if for each label $\ell$ it holds that $|S \cap
  C_\ell| \in b(\ell)$.
\end{definition}

If we have $m$ candidates then specifying independent constraints
requires providing at most $m+1$ numbers for each label.  Thus
independent constraints can easily be encoded in the inputs for our
algorithms.

Independent constraints are quite expressive. For example, they are
sufficient to express conditions such as ``the committee must contain
an even number of junior researchers'' or, since our committees are of
a given fixed size, conditions of the form ``the committee must
contain at least $40\%$ females.'' Indeed, the conditions of the
latter form are so important that we consider them separately.

\begin{definition}%[Independent constraints] 
  For a set of candidates~$C$, a set of labels~$L$, and a labeling
  function~$\lambda$, we say that a diversity specification~$D$ is
  interval-based (consists of interval constraints) if and only if
  there are functions $b_1, b_2 \colon L \rightarrow 2^{[|C|]}$
  (referred to as the lower and upper interval constraint functions)
  such that a committee $S$ is diverse if and only if for each label
  $\ell$ it holds that $b_1(\ell) \leq |S \cap C_\ell| \leq
  b_2(\ell)$.
\end{definition}

\paragraph{Label Structures}

In principle, our model allows each candidate to have an arbitrary set
of labels. In practice, there usually are some dependencies between
the labels and these dependencies can have strong impact in the
complexity of our problem. We focus on labels that are arranged in
independent, possibly hierarchically structured, layers.

Let $C$ be a set of candidates, let $L$ be a set of labels, and let
$\lambda$ be a labeling function.  We say that $\lambda$ has
\emph{$1$-layered} structure (i.e., we have a $1$-layered labeling) if
for each two distinct labels $x,y$ it holds that $C_x \cap
C_y=\emptyset$ (i.e., each candidate has at most one of these labels).
For example, if we restricted the example from the introduction to
labels regarding the seniority level (junior or senior), then we would
have a $1$-layered labeling.

More generally, we say that a labeling is \emph{$1$-laminar} if for
each two distinct labels $x,y$ we have that either (a) $C_x \cap
C_y=\emptyset$ or (b) $C_x \subseteq C_y$ or (c) $C_y \subseteq C_x$.
In other words, $1$-laminar labellings allow the labels to be arranged
hierarchically.

\begin{example}%[A multiwinner election with a regional cap]
  Consider a set $C=\{a,b,c,d,e\}$ of five candidates and labels that
  encode the countries and continents where the candidates come
  from. Specifically, there are four countries $r_1,r_2,r_3,r_4$, and
  two continents $R_1$ and $R_2$. The candidates are labeled as
  follows:
  \begin{align*}
    &\lambda(a)=\{r_1,R_1\}, \quad \lambda(b)=\{r_1,R_1\}, \quad \lambda(c)=\{r_2,R_1\}, \\
    &\lambda(d)=\{r_3,R_2\}, \quad \lambda(e)=\{r_4,R_2\}.
  \end{align*}
  \Cref{fig1} illustrates the $1$-laminar inclusion-wise relations between the
  labels (there can be more levels of the hierarchy; for example, for each country there could be
  labels specifying local administrative division).
  % Now, suppose that we want to select three candidates under the
  % constraint that there is at most one candidate from each country,
  % and there are at most two candidates from each continent. An
  % example of a diverse committee in this case is $\{a,c,e\}$.
\end{example}

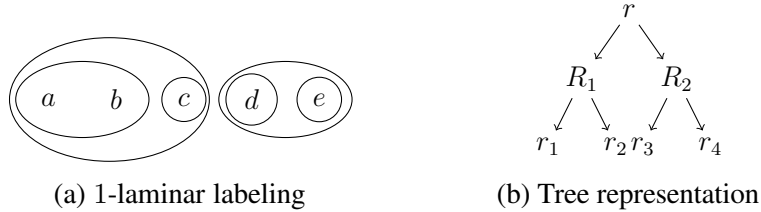
\begin{figure}
\centering
\begin{subfigure}[t]{0.37\textwidth}
\centering
\begin{tikzpicture}[scale=0.9, transform shape]
	\draw (0.5,0) ellipse (28pt and 16pt);
	\draw (0.9,0) ellipse (42pt and 26pt);
	\node at (0,0) {$a$};
	\node at (1,0) {$b$};
	\node[draw,circle] at (2,0) {$c$};
	
	\draw (3.5,0) ellipse (28pt and 16pt);
	\node[draw,circle] at (3,0) {$d$};
	\node[draw,circle] at (4,0) {$e$};
	%\draw [rounded corners=3mm]  (-3,-4)--(4,-4)--(0.5,2.5)--cycle;
\end{tikzpicture}
	\caption{$1$-laminar labeling}
	\label{fig1}
\end{subfigure}\qquad
\begin{subfigure}[t]{0.27\textwidth}
\centering
\begin{tikzpicture}[scale=0.9,transform shape]
		\tikzstyle{every node}=[inner sep=3pt] 
		\node(0) at (0,0) {$r$};
		\node(1) at (-0.7,-1) {$R_1$};
		\node(2) at (0.7,-1) {$R_2$};
		\node(3) at (-1.2,-2) {$r_1$};
		\node(4) at (-0.2,-2) {$r_2$};
		\node(5) at (0.2,-2) {$r_3$};
		\node(6) at (1.2,-2) {$r_4$};
		
		\draw[->] (0)--(1); 
		\draw[->] (0)--(2); 
		\draw[->] (1)--(3);
		\draw[->] (1)--(4); 
		\draw[->] (2)--(5);
		\draw[->] (2)--(6); 
\end{tikzpicture}
\caption{Tree representation}
\label{fig2}
\end{subfigure}
\caption{Illustration of a $1$-laminar labeling structure.}
\end{figure}

Every $1$-laminar labeling, together with the set of candidates, can
be represented as a rooted tree $T$ in the following way: For a pair
of distinct labels $x,y$ we create an arc from $x$ to $y$ if $C_x
\subsetneq C_y$ and there is no label $z$ such that $C_x \subsetneq
C_z \subsetneq C_y$. We add a \emph{root label} $r$ and we impose that
each candidate has this label; we add an arc from $r$ to each label
without an incoming arc. The resulting digraph $T$ is clearly a rooted
tree. See \Cref{fig2} for an illustration.

For each positive integer $t$, we say that a labeling is
\emph{$t$-layered} (respectively, \emph{$t$-laminar}) if the set $L$
of labels can be partitioned into sets $L_1, L_2, \ldots, L_t$ such
that for each $i \in [t]$, the labeling restricted to the labels from
$L_i$ is $1$-layered (respective, $1$-laminar). 

\begin{example}
  In the example from the introduction, restricting our attention to
  candidates' gender and seniority levels, we get a $2$-layered
  labeling structure. If we also consider labels regarding countries
  and continents, then we get a $3$-laminar structure (however, only
  the geographic labels would be using the full power of laminar
  labellings).
\end{example}

We assume that when we are given a $t$-layered ($t$-laminar) labeling
structure, we are also given the partition of the set of labels that
defines this structure (in \Cref{sec:structure-recognition} we
analyze the problem of recognizing such structures algorithmically).

\paragraph{Balanced Committee Model}

As a very natural special case of our model we considered the problem
of computing balanced committees. In this case there are only two
labels (e.g., \emph{male} and \emph{female}), each candidate has
exactly one label, and the constraint specification is that we need to
select exactly the same number of candidates with either label (thus,
by definition, the committee must be of an even size).

Computing balanced committees is a very natural problem. For example,
seeking gender balance is a common requirement in many settings. In
this paper, we seek exact balance (that is, we seek exactly the same
number of candidates with either label) but allowing any other
proportion would lead to similar results.

\section{Separable Objective Functions}\label{sec:separable}

Separable objective functions form a simple, but very important
special case of our setting. Indeed, such functions are very natural
in shortlisting examples, where diversity constraints are used to
implement, e.g., affirmative actions or employment-equity laws.  We
organize our discussion with respect to the type of constraint
specifications.

\paragraph{Independent Constraints}

It turns out that independent constraints are quite difficult to work
with.  If the labels are $1$-laminar then polynomial-time algorithms
exist (both for deciding if feasible committees exist and for
computing optimal ones), but with $2$-layered labellings our problems
become $\np$-hard (recall that $t$-layered labellings are a special
case of $t$-laminar ones).
Our polynomial-time algorithms proceed via dynamic programming and hardness
proofs use reductions from \textsc{Exact 3 Set Cover (X3C)}.

\begin{theorem}\label{thm:poly-independent-1laminar}
Let $D$ be a diversity specification of independent constraints.
Suppose that $\lambda$~is $1$-laminar and $f$~is separable.
Then, \DCWD{} can be solved in $O(|L|^2k^2 + |C|\log|C|))$ time.
Moreover, \DCF{} can be solved in $O(|L|^2k^2)$ time. 
If the $\lambda$ function is $2$-layered then both problems are
$\np$-hard (even if each candidate has at most two labels, and
 each label is associated to at most three candidates).
\end{theorem}
\begin{proof}
  We first consider the case where $\lambda$ is $1$-laminar and we give
  a polynomial-time algorithm.

  Let $b \colon L \rightarrow 2^{[|C|]}$ be the cardinality constraint
  function corresponding to diversity specification~$D$ of the input.
  Let $T$ be a rooted tree representation for $L$; we denote by $r$
  the root label that corresponds to the size $k$ constraint on the
  whole committee size, i.e., $C_r=C$ and $b(r)=\{k\}$.
  %As mentioned earlier, if such root label does not already exist, we simply add an artificial root label~$r$ with $b(r)=\{k\}$ and modify the labeling function ensuring that every candidate has this root label.
  Additionally, we add for every non-leaf label~$q$ in the tree
  representation of the labeling structure (including the possibly
  newly created~$r$) an artificial label~$q^*$ with~$b(q^*)=[0,|C|]$ and
  add this label to every candidate that has label~$q$ but none of the
  (original) child-labels of~$q$.
  This step clearly does no influence the solvability of our problem
  but ensures that every candidate has at least one label that is a
  leaf node in the tree representation of the labeling structure.

  For each $\ell \in L$, we denote by $\chld(\ell,i)$ the $i$th child
  of~$\ell \in L$ in~$T$, and by $\nbchld(\ell)$ number of children
  of~$\ell \in L$ in~$T$.
  By $\desc(\ell)$ we denote the set of all descendants of~$\ell$ (including $\ell$ itself, i.e., $\ell \in \desc(\ell)$).
  Furthermore, let~$\best(\ell,j)$ be the candidate from~$C_{\ell}$
  with the $j$th largest value according to~$f$.
  For technical reasons, we introduce the $\bot$~symbol as
  placeholder for a non-existing (sub)committee and define
  $X \cup \bot := \bot$ for any set~$X$.
  We set~$f(\emptyset):=0$ and~$f(\bot):=-\infty$.
  
  We describe a dynamic programming algorithm that solves \DCWD{}
  using the integer table~$\tblOpt$ where $\tblOpt[\ell,w,i]$
  contains a (sub)committee~$W$ with maximum total score~$f(W)$
  among all committees that consist of $|W|=w$~candidates with labels
  from~$\{\chld(\ell,1),\dots,\chld(\ell,i)\}$ such that $|W \cap
  C_{\ell'}| \in b(\ell')$ for all $\ell' \in \desc(\chld(\ell,j))$ and $j =1,2,\ldots,i$.

  It is not hard to verify that the overall solution for the \DCWD{}
  instance can be read from the table as~$\tblOpt[r,k,\nbchld(r)]$.

  We will now show how to compute the table~$\tblOpt$ in a bottom-up
  manner.  For each leaf label-node~$\ell$ we
  set~$\tblOpt[\ell,w,0]:=\{\best(\ell,j) \mid j \le w\}$ if $w \in
  b(\ell)$, that is, $\tblOpt[\ell,w,0]$ is the set of the
  $w$~``best'' candidates with label~$\ell$ and, otherwise, we
  set~$\tblOpt[\ell,w,0]:=\bot$.
  For each inner label-node~$\ell$, 
  we set~$\tblOpt[\ell,w,1]$ to $\tblOpt[\chld(\ell,i),x^*,\nbchld(\chld(\ell,i))]$
  where \[x^* := \argmax_{x \in [w]} f(\tblOpt[\chld(\ell,i),x,\nbchld(\chld(\ell,i))])\]
  if $w \in b(\ell)$ and, otherwise, we set~$\tblOpt[\ell,w,1]:=\bot$.
  Further, for each inner label-node~$\ell$ and $i>1$, 
  we set~$\tblOpt[\ell,w,i]$ to $\tblOpt[\ell,w-x^*,i-1]$ $\cup$
  $\tblOpt[\chld(\ell,i),x^*,\nbchld(\chld(\ell,i))]$
%   \[
%   \tblOpt[\ell,w-x^*,i-1] \cup \tblOpt[\chld(\ell,i),x^*,\nbchld(\chld(\ell,i))]
%   \]
  where $x^* := \argmax_{x \in [w]} f(\tblOpt[\ell,w-x,i-1] \cup
  \tblOpt[\chld(\ell,i),x,\nbchld(\chld(\ell,i))])$ if $w \in b(\ell)$
  and, otherwise, we set~$\tblOpt[\ell,w,i]:=\bot$.

  As for the running time, sorting the candidates with respect to
  their value according to~$f$ takes $O(|C| \cdot \log |C|)$ time.
  The table is of size~$O(|L|^2 k)$ and computing a single table
  entries takes at most~$O(k)$ time.  The overall running time is
  $O(|L|^2k^2 + |C|\log|C|))$ which is polynomial since~$k\le |C|$.

  For \DCF{}, we can skip to sort candidates which leads to the
  improved running time $O(|L|^2k^2)$.\bigskip

  Let us now consider the second part of the theorem.  We use a
  reduction from the $\np$-hard \textsc{Exact Cover by $3$-Sets}
  which, given a finite set~$X$ and a collection~$\mathcal{S}$ of
  size-$3$ subsets of~$X$, asks whether there is a subcollection
  $\mathcal{S}' \subseteq \mathcal{S}$ that partitions $X$, that is,
  each element~of $X$ is contained in exactly one subset
  from~$\mathcal{S}'$.  The reduction is similar to the reduction of
  the somewhat closely related {\sc General Factor}
  problem~\citep{Cor88} and works as follows: Create one \emph{element
    label}~$x$ for each element~$x \in X$ and one \emph{set label}~$S$
  for each subset~$S \in \mathcal{S}$.  We set $b(S):=\{0,3\}$ for
  each $S \in \mathcal{S}$ and $b(x):=\{1\}$ for each $x \in X$.  For
  each subset~$S=\{x,x',x''\} \in \mathcal{S}$, create three
  candidates $c(S,x)$, $c(S,x')$, and $c(S,x'')$ labeled with
  $\{S,x\}$, $\{S,x'\}$, and $\{S,x''\}$, respectively.  Finally, set
  the committee size~$k:=|X|$.  This completes the construction which
  can clearly be performed in polynomial time.  For the correctness,
  assume that there is a subcollection $\mathcal{S}' \subseteq
  \mathcal{S}$ that partitions~$X$.  It is easy to verify that
  $\{c(S,x^*) \mid S \in \mathcal{S}', x^* \in S\}$ is a $D$-diverse
  committee.  Furthermore, let $C^* \subseteq C$ be an arbitrary
  $D$-diverse committee.  Now, ${\cal S}'=\{S \in {\cal S} \mid c(S,x) \in C^* \text{ for
    some $x \in S$}\}$ partitions~$X$: each element~$x \in X$ is
  covered exactly once since $b(x)=\{1\}$ for all $x \in X$, and ${\cal S}'$ is pairwise disjoint since
  $b(S)=\{0,3\}$ for all $S \in \mathcal{S}$.
\end{proof}

Given the above hardness results, it is immediate to ask about the
parametrized complexity of our problems because in many settings the
label structures are very limited (for example, the $2$-layered
gender/seniority labeling from the introduction contains only $4$
labels and already is very relevant for practical applications).
Unfortunately, for independent constraints our problems remain hard
when parametrized by the number of labels.

\begin{theorem}\label{thm:W1h_labels-independent}
  Both \DCF{} and \DCWD{} problems are $W[1]$-hard with respect to the
  number of labels $|L|$, even if $D$ is a diversity specification of
  independent constraints.
\end{theorem}

\begin{proof}
We describe a parametrized reduction from the $W[1]$-hard {\sc Multicolored Clique} ({\sc MCC}) problem which,
given an undirected graph~$G=(V,E)$, a non-negative integer~$h\in \naturals$,
and a vertex coloring~$\phi \colon V\to \{1, 2, \ldots, h\}$, asks
whether graph~$G$ admits a colorful $h$-clique, that is, a size-$h$ vertex
subset~$H\subseteq V$ such that the vertices in~$H$ are pairwise adjacent
and have pairwise distinct colors.
Without loss of generality, we assume that the number of vertices from each color class
equals some integer~$q \le |V|$.
Let $(G=(V,E),\phi)$ be an {\sc MCC} instance.
We denote the set of vertices of color~$i$ as $V(i) = \{v_1^{i}, \ldots, v_q^{i}\}$.
%and denote the set of edges that connect a vertex of color~$i$ to a vertex of color~$j$
%as $E(i,j) = \{e_1^{i,j}, \ldots, e_{|E(i,j)|}^{i,j}\}$.
We construct a \DCF{} instance as follows.

\emph{Labels.} \quad
For each color~$i \in [h]$ we have a \emph{lower vertex label}~$\vertllb_i$ and
a higher \emph{higher vertex label}~$\verthlb_i$.
For each (unordered) color pair~$i,j \in [h], i \neq j$ we have an \emph{edge label}~$\edgelb_{i,j}$.
(So that $|L|=2h+h(h-1)/2$ is obviously upper-bounded by some function in~$h$.)

\emph{Candidates and Labeling.} \quad
For each color~$i \in [h]$ and each vertex~$v \in V(i)$ we introduce $q(q+1)$~\emph{lower color-$i$-selection candidates} and
$q$~\emph{higher color-$i$-selection candidates}.
The labeling function~$\lambda$ is defined as follows.
For each lower color-$i$-selection candidate~$c$ we have
$\lambda(c):=\{\vertllb_i, \verthlb_i\} \cup \{\edgelb_{i,j} \mid j < i\}$.
For each higher color-$i$-selection candidate~$c$ we have
$\lambda(c):=\{\verthlb_i\} \cup \{\edgelb_{i,j} \mid j > i\}$.
Introduce further $h(q+2)^2$~\emph{dummy candidates}~$u$ with $\lambda(u):=\emptyset$.

\emph{Diversity Constraints.} \quad
We define the cardinality constraint function~$b$ as follows.
For each color~$i \in [h]$ we set~$b(\vertllb_i):=\{(q+1)x \mid 1 \le x \le q\}$
and set~$b(\verthlb_i):=\{y \mid 1 \le y \le q\}$.
For each (unordered) color pair~$i,j \in [h], i<j$ we set 
$b(\edgelb_{i,j}):=\{(q+1)x+y \mid$ there is an edge between the $x$th vertex from~$V(i)$
and the $y$th vertex from~$V(j)\}$.

We finally set the committee size~$k:=h(q+2)^2$.
This completes the reduction which clearly runs in polynomial time.
It remains to show that the graph~$G$ has a colorful $h$-clique if and only if
the constructed \DCF{} instance admits a diverse committee.

Assume that $G$~has a colorful $h$-clique~$H$.
Let $\id(H,i)$ denote the index of the color~$i$ vertex from~$H$,
that is, $\id(H,i)=x$ if and only if $H$~contains the $x$th vertex of color~$i$.
It is not hard to verify that a diverse committee can be constructed as follows.
Start with a committee that consists only of $k$~dummy candidates.
For each color~$i \in [h]$ replace $(q+1)\id(H,i)$~dummy candidates by lower color~$i$-selection candidates
and replace $\id(H,i)$~dummy candidates by higher color~$i$-selection candidates.
The diversity constraints of the lower and higher vertex labels are clearly fulfilled by this construction.
Now, consider some edge label~$\edgelb_{i,j}, i<j$.
Our construction ensures that there are exactly~$\id(H,i) (q+1)$ lower color-$i$-selection candidates in the committee
with label~$\edgelb_{i,j}$ and further $\id(H,j)$~higher color-$j$-selection candidates with label~$\edgelb_{i,j}$
(and no further candidates with label~$\edgelb_{i,j}$).
Since $H$~is a clique, we know that the $\id(H,i)$th vertex of color~$i$ is adjacent to he $\id(H,j)$th
vertex of color~$j$ and thus $\id(H,i) (q+1) + \id(H,j) \in b(\edgelb_{i,j})$.
Thus, also the diversity constraints for the edge labels are fulfilled and the committee is indeed diverse.

Finally, assume that the constructed \DCF{} instance admits some diverse committee.
To fulfill the diversity constraints for the lower vertex labels for each color~$i \in [h]$ there is some
number~$\id(i)$ such that there are exactly $(q+1)\id(i)$~lower color-$i$-selection candidates
and further $\id(i)$~higher color-$i$-selection candidates in the committee.
(The former is directly enforced by the diversity constraints for the lower vertex labels
and the latter follows then immediately from the diversity constraints for the higher vertex labels.)
We claim that~$H=\{v^H_i \mid \text{vertex~$v_i$ is the $\id(i)$th vertex of color~$i$}\}$ is an
$h$-colored clique.
It is clear from the definition of~$H$ that $|H|=h$ and that~$H$ is $h$-colored but it remains to show
that $H$~is indeed a clique.
To show this, suppose towards a contradiction that there are two colors~$i,j \in [h], i<j$ such that
vertex~$v^H_i$ and vertex~$v^H_j$ are not adjacent.
Now, there are exactly~$\id(i) (q+1)$ lower color-$i$-selection candidates in the committee
with label~$\edgelb_{i,j}$ and further $\id(j)$~higher color-$j$-selection candidates with label~$\edgelb_{i,j}$
(and no further candidates with label~$\edgelb_{i,j}$).
Furthermore, since the diversity constraint of label~$\edgelb_{i,j}$ is fulfilled, it must hold that
$\id(i) (q+1) + \id(j) \in b(\edgelb_{i,j})$ and so that vertex~$v^H_i$ and vertex~$v^H_j$ are
adjacent---a contradiction.
\end{proof}

% Hardness of \DCF{} clearly transfers to \DCWD{}
% (even for separable rules).

% \begin{corollary}
%  \DCWD{} is $\np$-hard even if $D$ is a diversity specification of independent constraints, $L$~is two-layered, each candidate has at most two labels, and
%  each label is associated to at most three candidates, and it is $W[1]$-hard with respect to the number of labels~$|L|$.
% \end{corollary}

However, not all is lost and sometimes brute-force algorithms are
sufficiently effective. For example, if we have a $t$-layered labeling
(where $t$ is a small constant) then 
each candidate has at most $t$
different labels and it suffices to consider each size-$t$ labeling
separately.
% . For each set $A$ of at most $t$ labels we guess the
% number $n_A$ of candidates with exactly this labeling, for each
% different set $A$ of $t$ labels we choose $n_A$ candidates with
% exactly labeling $A$ and with highest weights; among such committees
% we choose one with the highest objective value (that is of required
% size and satisfies the diversity constraints).
% %
A brute-force algorithm based on this idea suffices, e.g., for the
example from the seniority/specialty labels from the introduction (it
would have $O(|C|^4)$ running time, because there are $4$ combinations
of labels $\{\mathit{junior}, \mathit{senior}\}$ and $\{\mathit{AI},
\mathit{economics}\}$; the algorithm could also deal with
non-independent constraints).

\paragraph{Interval constraints}
% We observe that feasibility checking can be done in $O(|C|+|L|)$ time if the set of interval constraints has a nested structure. 
% %We note that if $L$ is 1-laminar, then the family $\{\, C_{\ell} \mid \ell \in L \,\}$ has size at most $2|C|-1$ \cite{Korte2006}; hence without loss of generality, we can assume that $|L|=O(|C|)$.

Interval constraints are more restrictive than general independent
ones, but usually suffice for practical applications and are more
tractable.  For example, for the case of $1$-laminar labellings we
give a linear-time algorithm for recognizing if a feasible committee
exists (for independent constraints, our best algorithm for this task
is quadratic).

\begin{theorem}\label{thm:feasibility:1laminar}
Let $D$ be a diversity specification of interval constraints. If $\lambda$ is $1$-laminar, then \DCF{} can be solved in $O(|C|+|L|)$ time. 
\end{theorem}

\begin{proof}
Let $b_1$ and $b_2$ denote the lower and upper interval constraint functions, respectively.
Let $T$ be a rooted tree representation for $L$; we denote by $r$ the root label that corresponds to the size $k$ constraint on the whole committee size, i.e., $C_r=C$ and $b_1(r)=b_2(r)=k$. For each $\ell \in L$, we denote by $\ch(\ell)$ the set of children of $\ell$ and by $\desc(\ell)$ the set of descendants of $\ell$ in $T$ (including~$\ell$). 

For every label $\ell \in L$, let $\calK_{\ell}$ be the set of committees $W \subseteq C_{\ell}$ satisfying the constraints up until $\ell$, namely, $b_1(x) \leq |W \cap C_{x}| \leq  b_2(x)$ for each $x \in \desc(\ell)$; then, we define $A_1[\ell]$ (respectively, $A_2[\ell]$) to be the minimum (respectively, the maximum) value $w$ for which there is a set $W \in \calK_{\ell}$ with $|W|=w$. If there exists no committee satisfying the aforementioned constraints we set $A_1[\ell]$ and $A_2[\ell]$ to $\infty$ and $-\infty$, respectively. Clearly, there is a $D$-diverse committee if and only if $A_1[r]=A_2[r]=k$. The values $A_1[\ell]$ and $A_2[\ell]$ can be efficiently computed by a dynamic programming in a bottom-up manner as follows. 

For each leaf $\ell \in L$ at $T$, we set $A_1[\ell]=b_1(\ell)$ and $A_2[\ell]=\min\{b_2(\ell), |C_{\ell}|\}$ if $b_1(\ell) \leq \min \{b_2(\ell),|C_{\ell}|\}$; we set $A_1[\ell]=+\infty$ and $A_2[\ell]=-\infty$ otherwise. For each internal node $\ell \in L$, we set $A_1[\ell]=\max \{\sum_{x \in \ch(\ell)}T_1(x),b_1(\ell)\}$ and $A_2[\ell]= \min \{\sum_{x \in \ch(\ell)} A_2[x],b_2(\ell),|C_{\ell}|\}$ if 
\begin{itemize}
\item $A_1[x] \leq A_2[x]$ for all child $x \in \ch(\ell)$; and
\item there are enough candidates in $C_{\ell}$ to fill in the lower bound ($|C_{\ell}| \ge A_1[\ell]$) and the lower bound does not exceed the upper bound, i.e., 
\[\max \{\sum_{x \in \ch(\ell)}A_1[x],b_1(\ell)\}\leq \min \{\sum_{x \in \ch(\ell)} A_2[x],b_2(\ell),|C_{\ell}|\}.\]
\end{itemize}
Otherwise, we set $A_1[\ell]=+\infty$ and $A_2[\ell]=-\infty$. 
%Running time
This can be done in $O(|C|+|L|)$ time since each $|C_{\ell}|$ for $\ell \in L$ can be computed in $O(|C|)$ time and since the size of the dynamic programming table is at most $|L|$ and each entry can be filled in constant time.
%It can be shown that for each $\ell \in L$ and each $w \in \naturals$, there is a set $W \in \calK_{\ell}$ of size $w$ if and only if $A_1[\ell] \leq w \leq A_2[\ell]$.

% \iffalse
% TODO decide whether the first part of the proof should be moved back to main text
% if so, reinclude the text below:
% "; see the supplementary material for a proof."
%Correctness
Now we will show by induction that for each $\ell \in L$ and each $w \in \naturals$, there is a set $W \in \calK_{\ell}$ of size $w$ if and only if $A_1[\ell] \leq w \leq A_2[\ell]$. The claim is immediate when $\ch(\ell)=\emptyset$. Now consider an internal node $\ell \in L$ and suppose that the claim holds for all $x \in \ch(\ell)$. 

Suppose first that $A_1[\ell] \leq w \leq A_2[\ell]$. By induction hypothesis, for each child $x \in \ch(\ell)$, there is a committee $W_x \in \calK_{x}$ where $|W_x|=w_x$ for any $w_x \in [A_1[x],A_2[x]]$. By combining all such committees, we have that for any $t \in [\sum_{x \in \ch(\ell)}A_1[x], \sum_{x \in \ch(\ell)}T_2(x)]$, there is a committee $W \subseteq C_{\ell}$ of size $t$ such that $W \cap C_{x} \in \calK_{x}$ for all $x \in \ch(\ell)$. In particular, since $w \in [\sum_{x \in \ch(\ell)}A_1[x], \sum_{x \in \ch(\ell)}T_2(x)]$, there is a set $W \subseteq C_{\ell}$ of size $w$ such that $W \cap C_{x} \in \calK_{x}$ for all $x \in \ch(\ell)$. Since $b_1(\ell) \leq w \leq b_2(\ell)$, we have $W \in \calK_{\ell}$. 

Conversely, suppose that $w$ does not belong to the interval $[A_1[\ell],A_2[\ell]]$. Suppose towards a contradiction that there is a set $W \in \calK_{\ell}$ of size $w$. Notice that for each $x \in \ch(\ell)$, it holds that $W \cap C_x \in \calK_x$ and hence $A_1[x] \leq |W \cap C_x| \leq A_2[x]$ by induction hypothesis. If $w<b_1(\ell)$ or $w>b_2(\ell)$, it is clear that $W \not \in \calK_{\ell}$, a contradiction. Further, if $w>|C_{\ell}|$, $W$ cannot be a subset of $C_{\ell}$, a contradiction. If $w<\sum_{x \in \ch(\ell)}A_1[\ell]$, then $|W \cap C_x| <A_1[x]$ for some label $x \in \ch(\ell)$; however, since $W \cap C_x \in \calK_x$, we have $|W \cap C_x| \geq A_1[x] $ by induction hypothesis, a contradiction. A similar argument leads to a contradiction if we assume $w>\sum_{x \in \ch(\ell)}A_1[\ell]$. 
% \fi
\end{proof}

For the case of computing the winning committee we no longer obtain a
significant speedup from focusing on interval constraints, but we do
get a much better structural understanding of the problem. In
particular, we can use a greedy algorithm instead of relying on
dynamic programming.
Briefly put, our algorithm (presented as \Cref{alg:greedy})
starts with an empty committee and performs $k$ iterations ($k$ is the
desired committee size), in each extending the committee with a
candidate that increases the score maximally, while ensuring that the
committee can still be extended to one that meets the diversity
constraints.
To show that this greedy algorithm is correct and that it can be
implemented efficiently, we use some notions from the matroid theory.
% In the non-constrained setting, the greedy algorithm finds an optimal
% solution for every separable function. A natural question is for which
% constraint structures does the greedy algorithm produce an optimal
% solution? This question can be answered by classical results in
% matroid theory. 

Formally, a \emph{matroid} is an ordered pair $(C,\calI)$, where $C$
is some finite set and $\calI$ is a family of its subsets (referred to
as the \emph{independent sets} of the matroid). We require that
(I1)~$\emptyset \in \calI$, 
(I2) if $S \subseteq T \in \calI$, then $S \in \calI$, and 
(I3) if $S,T \in \calI$ and $|S| > |T|$, then there exists $s \in S \setminus T$ such that $T\cup \{s\} \in
\calI$. 
The family of maximal (with respect to inclusion) independent sets of
a matroid is called its \emph{basis}.
Many of our arguments use results from matroid theory, but often used
in very different contexts than originally developed. In particular,
the next theorem, in essence, translates the results of
\citet{Yokoi2017} to our setting.

% Formally, the ordered pair $(C,\calI)$ that consists of a finite set
% $C$ and a family of subsets $\calI \subseteq 2^C$, is called a {\em
%   matroid} if it satisfies the following three axioms: $($I$1)$
% $\emptyset \in \calI$, $($I$2)$ if $S \subseteq T \in \calI$, then $S
% \in \calI$, and $($I$3)$ if $S,T \in \calI$ and $|S| > |T|$, then
% there exists $s \in S \setminus T$ such that $T\cup \{s\} \in
% \calI$. We call $\calI$ the {\em independent sets} of a
% matroid. Maximal independent sets are called a {\em basis}. In what
% follows, we will see that our problems can be put into the framework
% of matroids. The following theorem can be implied by known results in
% a different context (see Proposition $2$ of \cite{Yokoi2017}).

\begin{theorem}\label{DCWD-separable-1laminar}
  Let $D$ be a diversity specification of interval
  constraints. Suppose that $\lambda$ is a 1-laminar, and $f$ is a
  separable function given by a weight vector $\bfw\colon C
  \rightarrow \reals$. Then, \DCWD{} can be solved in
  $O(k^2|C||L|+|C|\log |C|)$ time.
\end{theorem}

\begin{algorithm}[t]
\SetKwInOut{Notation}{notation} \SetKwInOut{Output}{output}
\SetKwInOut{Input}{input} 
\SetKw{And}{and}
\SetKw{None}{None}
\caption{Greedy Algorithm $1$}\label{alg:greedy}
\Notation{$\calK_D \neq \emptyset$ is the set of $D$-diverse, size-$k$ committees,
$\overline \calK_D$ is its lower extension.
}
\Input{$f\colon2^C \rightarrow \reals$: the objective function, \\
$k$: the size of the committee.
}
\Output{$W \in \calK_D$
}
	set $W=\emptyset$\;
	\While{$|W|<k$}{
	choose a candidate $y \in C\setminus W$ such that $W\cup\{y\} \in {\overline \calK}_D$ with the maximum improvement $f(\{c\}|W)$\;
	set $W \leftarrow W \cup \{y\}$\;
	}
\end{algorithm}

\begin{proof}
Let $\calK_D$ be the set of $D$-diverse committees of size $k$ and assume that $\calK_D$ is nonempty. For a family of subsets $\calK$ of a finite set $C$, we define its lower extension\footnote{Note that the lower extension does not necessarily ignore the lower bounds. For instance, consider when we want to select a committee of size $3$ such that there are exactly three female candidates and at most two male candidates; the corresponding lower extension ${\overline \calK}_D$ only includes the sets of female candidates of size at most $3$, whereas a male-only committee of size $2$ satisfies the upper bounds.} by
\[
{\overline \calK}=\{\, T \mid \exists S \in \calK: T \subseteq S \,\}.
\]
It is known that if our constraints are given by intervals, the lower extension ${\overline \calK}_D$ of $\calK_D$ comprises the independent sets of a matroid whenever $\calK_D \neq \emptyset$ \citep{Yokoi2017}. Thus, the greedy algorithm (\Cref{alg:greedy}) finds an optimal solution $W \in \argmax_{W^{\prime} \in {\overline \calK}_D} f(W^{\prime})$ (see, e.g., the book of \citet{Korte2006}, Chapter~13). By construction, $|W|=k$ and hence $W$ is a maximal element in ${\overline \calK}_D$, which follows that $W\in \calK_D$. Since $\calK_D \subseteq {\overline \calK}_D$, we have $W \in \argmax_{W^{\prime} \in {\calK}_D} f(W^{\prime})$. Further, \citet{Yokoi2017} showed that checking whether a set $W\cup \{y\}$ belongs to ${\overline \calK}_D$ can be efficiently done by maintaining a set $B \in \calK_D$ with $W \subseteq B$; thus, the greedy algorithm runs in polynomial time
% ; we provide a full proof in the supplementary material. 

Now it remains to analyze the running time of the algorithm. Sorting the candidates from the best to the worst requires $O(|C|\log |C|)$ time, given a weight vector $\bfw$. In each step, we need to check whether a set $W\cup \{y\}$ belongs to ${\overline \calK}_D$. \citet{Yokoi2017} showed that this can be efficiently done by maintaining a set $B \in \calK_D$ with $W \subseteq B$. Specifically, the following lemma holds. 
\begin{lemma}[Lemma $6$ of \citealp{Yokoi2017}]\label{lem:membership}
Let $(C,\calI)$ be a matroid. Let $W$ be an independent set of the matroid, $B$ be a basis with $W \subseteq B$, and $y \in C\setminus W$. Then, $W \cup \{y\}$ is independent if and only if $y \in B$ or $(B \cup \{y\}) \setminus \{x\}$ is a basis for some $x \in B \setminus W$.
\end{lemma}
The lemma implies that provided a set $B \in \calK_D$ with $W \subseteq B$, deciding $W\cup \{y\} \in {\overline \calK}_D$ can be verified by checking whether $y \in B$ or $(B\cup \{y\})\setminus \{x\} \in \calK_D$ for some $x \in B\setminus W$; this can be done in $O(k^2|L|)$ time. One can maintain such a superset $B \in {\overline \calK}_D$ of $W$ by first computing a set $B \in \calK_D$ in $O(|C|+|L|)$ time as we have proved in \Cref{thm:feasibility:1laminar}, and updating the set $B$ in each step as follows: If $y \not \in B$, then find a candidate $x \in B\setminus W$ such that $(B \cup \{y\}) \setminus \{x\}  \in \calK_D$, and set $B=(B \cup \{y\}) \setminus \{x\}$; otherwise, we do not change the set $B$.  
Since there are at most $|C|$ iterations, the greedy algorithm runs in $O(k^2|C||L|)$ time. 
\end{proof}

Unfortunately, the greedy algorithm does not work for more involved
labeling structures, but for $2$-laminar labellings we can compute
winning committees by reducing the problem to the matroid intersection
problem \citep{Edmonds1979}. For more involved labeling structures our
problems become $\np$-hard.

\begin{theorem}\label{thm:2laminar}
Let $D$ be a diversity specification of interval constraints. Suppose that $\lambda$ is 2-laminar and $f$ is separable.
Then, \DCF{} can be solved in $O(k^2|C|^3|L|)$ time, and \DCWD{} can be solved in $O(k|C|^3+k^3|C|^2|L|)$ time.
\end{theorem}

In the subsequent proof, we will use the following notions and results in matroid theory: Given a matroid $(C,\calI)$, the sets in $2^C \setminus \calI$ are called {\em dependent}, and a minimal dependent set of a matroid is called {\em circuit}. Crucial properties of circuits are the following. 
\begin{lemma}
Let $(C,\calI)$ be a matroid, $W \in \calI$, and $y \in C \setminus W$ such that $W \cup \{y\} \not \in \calI$. Then the set $W \cup \{y\}$ contains a unique circuit.
\end{lemma}
We write $C(W,y)$ for the unique circuit in $W\cup \{y\}$. The set $C(W,y)$ can be characterized by the elements that can replace $y$, i.e., for each independent set $W$ of a matroid $(C,\calI)$ and $y \in C \setminus W$ with $W \cup \{y\} \not \in \calI$, 
\[
C(W,y)=\{\, x \in W \cup \{y\} \mid W\cup \{y\} \setminus \{x\} \in \calI \,\}.
\]

The following lemma by \citet{Frank1981} serves as a fundamental property for proving the matroid intersection theorem. 

\begin{lemma}[\citealp{Frank1981}]\label{lem:exchange}
Let $(C,\calI)$ be a matroid and $W \in \calI$. Let $x_1,x_2, \ldots,x_s \in W$ and $y_1,y_2,\ldots,y_s \not \in W$ where $W\cup \{y_j\} \not \in \calI$ for $j \in [s]$. Suppose that 
\begin{itemize}
\item[$(${\rm i}$)$] $x_j \in C(W,y_j)$ for $j \in [s]$ and
\item[$(${\rm ii}$)$] $x_j \not \in C(W,y_t)$ for $1\leq j <t <s$.
\end{itemize}
Then, $(W \setminus \{x_1,x_2,\ldots,x_s\})\cup \{y_1,y_2,\ldots,y_s\} \in \calI$.
\end{lemma}

Now we are ready to prove \Cref{thm:2laminar}. 

\begin{proof}
Let $b_1$ and $b_2$ denote the lower and upper interval constraint functions, respectively.
Let $L=L_1 \cup L_2, L_1 \cap L_2 = \emptyset$ be a partition of~$L$ such
that for each $i=1,2$, the labeling restricted to the labels from $L_i$ is $1$-laminar. 
For $i=1,2$, we denote by $\calK_i$ the set of committees of size $k$ satisfying the constrains in $L_i$, i.e., $\calK_i=\{\, S \subseteq C \mid |S|=k~\mbox{and}~b_1(\ell) \leq |S \cap C_{\ell}| \leq b_2(\ell)~\mbox{for all}~\ell \in L_i\,\}$. If at least one of them is empty, then there is no $D$-diverse committee; thus we assume otherwise. 
We have argued that the lower extension $\overline \calK_i$ for each $i =1,2$ forms the independent sets of a matroid when $\lambda|_{L_i}$ is $1$-laminar. Thus, our problem can be reduced to finding a maximum common independent set over the two matroids. That is, we will try to compute the following value:
\[
\max \{\, |W| \mid W \in {\overline \calK_1} \cap {\overline \calK_2}\,\}. 
\]
Clearly, there is a $D$-diverse committee of size $k$ if and only the maximum value equals $k$. It is well-known that this problem can be solved by Edmond's matroid intersection algorithm \citep{Edmonds1979}, given a membership oracle for each $\overline \calK_i$. The idea is that starting with the empty set, we repeatedly find `alternating paths' and augment $W$ by one element in each iteration while keeping the property $W \in {\overline \calK}_1 \cap {\overline K}_2$. Specifically, we apply the notion $C(W,y)$ to $(C,{\overline \calK}_i)$ and write $C_i(W,y)$ for each $i=1,2$.
For $W \in {\overline \calK_1} \cap {\overline \calK_2}$, we define an auxiliary graph $G_W=(C,A^{(1)}_W\cup A^{(2)}_W)$ where the set of arcs is given by
\begin{align*}
&A^{(1)}_W= \{\, (x,y)\mid W\cup \{y\} \not \in {\overline \calK}_1 \land x \in C_1(W,y) \,\},\\
&A^{(2)}_W= \{\, (y,x)\mid W\cup \{y\} \not \in {\overline \calK}_2 \land x \in C_2(W,y) \,\},
\end{align*}
for $i=1,2$. We then look for a shortest path from $S^{(1)}_W$ to $S^{(2)}_W$, where 
\[
S^{(i)}_W=\{\, y \in C\setminus W \mid W \cup \{y\} \in \overline \calK_i \,\},
\]
for $i=1,2$. We increase the size of $W$ by taking the symmetric difference with the path. It was shown that this procedure computes the desired value. We provide a formal description of the algorithm below (\Cref{alg:intersection}). 

% TODO: Reinsert the following if the first part is reinserted to the main text.
%; see the supplementary material for a proof.

Similarly to \citet{Yokoi2017}, we can efficiently construct an auxiliary graph in each step by maintaining a set $B_i$ such that $W \subseteq B_i$ and $B_i \in \calK_i$ for each $i=1,2$: First, as we have seen in \Cref{lem:membership}, we can determine the membership of a given set in $\overline \calK_i$ in polynomial time. Moreover, it can be easily verified that the unique circuit $C_i(W,y)$ coincides with $C_i(B_i,y)$ when $W\cup \{y\} \not\in {\overline \calK}_i$.
\begin{lemma}\label{lem:circuits}
Let $(C,\calI)$ be a matroid. Let $W$ be an independent set of the matroid, $B$ be a basis with $W \subseteq B$, and $y \in C\setminus W$ with $W\cup \{y\}$ being dependent. Then, $C(W,y)=C(B,y)$.
\end{lemma}
\begin{proof}
Notice that $B\cup \{y\}$ is dependent: thus it contains a unique circuit $C(B,y)$. Then, $C(B,y)=C(W,y)$ clearly holds since $C(W,y) \subseteq B\cup \{y\}$.
\end{proof}

Thus, we will show how to maintain such a set $B_i \in  \calK_i$ with $W  \subseteq B_i$ for each $i=1,2$. We will first compute a set $B_i \in \calK_i$ for each $i=1,2$ in $O(|C||L|)$ time. Now suppose that $W\in {\overline \calK_1} \cap {\overline \calK_2}$. Let $B_i \in \calK_i$ where $W\subseteq B_i$ for $i=1,2$, and $P=(y_0,x_1,y_1,\ldots,x_s,y_s)$ be a shortest path in $G_W$ with $y_0 \in S^{(1)}_W$ and $y_s \in S^{(2)}_W$. Notice that $y_j \not \in B_1$ for $j=1,2,\ldots,s$ since otherwise $B_1$ contains a dependent set $W \cup \{y_j\}$ of a matroid $(C,\calK_1)$, contradicting $($I2$)$; similarly, $y_j \not \in B_2$ for $j=s,s-1,\ldots,1$ since otherwise $B_2$ contains a dependent set $W \cup \{y_j\}$ of a matroid $(C,\calK_2)$, a contradiction. If $y_0 \not \in B_1$, then there is a candidate $x \in B_1 \setminus W$ such that $(B_1\cup \{y_0\})\setminus \{x\} \in \calK_1$ by \Cref{lem:membership}, and we set $B_1$ to be $(B_1\cup \{y_0\})\setminus \{x\}$. Similarly, if $y_s \not \in B_2$, then there is a candidate $x \in B_2 \setminus W$ such that $(B_2\cup \{y_s\})\setminus \{x\} \in \calK_2$ by \Cref{lem:membership}, and we set $B_2$ to be $(B_2\cup \{y_0\})\setminus \{x\}$. We then update each $B_i$ as follows.  
\[
B^{\prime}_i=(B_i\cup \{y_0,y_1,\ldots,y_s\})\setminus \{x_1,x_2,\ldots,x_s\}. 
\]
Clearly, $W\cup \{y_0,y_1,\ldots,y_s\})\setminus \{x_1,x_2,\ldots,x_s\} \subseteq B^{\prime}_i$ for each $i=1,2$ since $y_0,y_1,\ldots,y_s \not \in W$ and $x_1,x_2,\ldots,x_s \in W$; further we have the following.

\begin{claim}
For $i=1,2$, $B^{\prime}_i \in \calK_i$.
\end{claim}
\begin{proof}
First, we show that $B_1\cup \{y_0\}$, $y_1,y_2,\ldots,y_s$, and $x_1,x_2,\ldots,x_s$ satisfy the requirements of \Cref{lem:exchange}. Since $y_0 \in B_1$, we know that $B_1=B_1\cup \{y_0\} \in \calK_1$. The condition $(${\rm i}$)$ is satisfied because $(x_j,y_j) \in A^{(1)}_W$ and $C_1(W,y_j)=C_1(B_1,y_j)$ for all $j=1,2,\ldots,s$. The condition $(${\rm ii}$)$ is satisfied because otherwise $x_j \not \in C(B,y_t)=C(W,y_t)$ for some $j<t$ and hence the path could be shortcut. Thus, $B^{\prime}_1 \in {\overline \calK}_1$. Moreover,  since $|B_1|=|B^{\prime}_1|=k$, $B^{\prime}_1$ is a maximal set in ${\overline \calK}_1$ and hence $B^{\prime}_1 \in \calK_1$. Similarly, one can show that $B_2\cup \{y_s\}$, $y_{s-1},y_{s-2},\ldots,y_0$, and $x_s,x_{s-1},\ldots,x_1$ satisfy the requirements of \Cref{lem:exchange}, and therefore $B^{\prime}_2 \in \calK_2$.
\end{proof}
It remains to analyze the running time of the algorithm. By \Cref{lem:membership} and \Cref{lem:circuits}, constructing an auxiliary graph $G_W$ can be done in $O(k^2|C|^2|L|)$ time, given that we have a superset $B_i$ of $W$ where $B_i \in \calK_i$, and finding a shortest path can be done in $O(|C|)$ time by breadth first search. Since there are at most $|C|$ augmentations, the overall running time is $O(k^2|C|^3|L|)$ time.
\end{proof}

\begin{algorithm}[t]
\SetKwInOut{Input}{input} \SetKwInOut{Output}{output}
\SetKw{And}{and}
\SetKw{None}{None}
\caption{Matroid intersection}\label{alg:intersection}
\Input{$\calK_i \neq \emptyset$ for $i=1,2$.
}
\Output{$W \in {\overline \calK_1} \cap {\overline \calK_2}$ of maximum cardinality
}
	set $W=\emptyset$\;
	compute $B_i \in \calK_i$ for each $i=1,2$\;
	\While{there is a shortest path $P=(y_0,x_1,y_1,\ldots,x_s,y_s)$ from $S^{(1)}_W$ to $S^{(2)}_W$ in $G_W$}{

	\If{$y_0 \not \in B_1$}{
	find $x \in B_1 \setminus W$ such that $(B_1\setminus \{x\})\cup \{y_0\} \in \calK_1$\;
	set $B_1=(B_1\setminus \{x\})\cup \{y_0\}$\;
	}
	\If{$y_s \not \in B_2$}{
	find $x \in B_2 \setminus W$ such that $(B_2\setminus \{x\})\cup \{y_s\} \in \calK_2$\;
	set $B_2=(B_2\setminus \{x\})\cup \{y_s\}$\;
	}
	 set $W=W \cup \{y_0,y_1,\ldots,y_s\})\setminus \{x_1,x_2,\ldots,x_s\}$ and $B_i=B_i \cup \{y_0,y_1,\ldots,y_s\})\setminus \{x_1,x_2,\ldots,x_s\}$ for each $i=1,2$\;
	}
\end{algorithm}

\begin{theorem}
Let $D$ be a diversity specification of interval constraints. Suppose that $\lambda$ is 2-laminar and $f$ is a separable function given by a weight vector $\bfw\colon C \rightarrow \reals$. Then, \DCWD{} can be solved in $O(k|C|^3+k^3|C|^2|L|)$ time.
\end{theorem}
\begin{proof}
Again, for $i=1,2$, we denote by $\calK_i$ the set of committees of size $k$ satisfying the constrains in $L_i$, i.e., $\calK_i=\{\, S \subseteq C \mid |S|=k~\mbox{and}~b_1(\ell) \leq |S \cap C_{\ell}| \leq b_2(\ell)~\mbox{for all}~\ell \in L_i\,\}$. We assume that $\calK_1 \cap \calK_2$ is nonempty. \citet{Frank1981} has shown that for each $k^{\prime} \in [k]$, one can calculate $W_{k^{\prime}} \in {\overline \calK_1} \cap {\overline \calK_2}$ where
\[
f(W_{k^{\prime}})= \max \{\, f(W) \mid W \in {\overline \calK_1} \cap {\overline \calK_2} \land |W|=k^{\prime}\,\}
\]
in $O(|C|^3+\gamma)$, where $\gamma$ is the time required for constructing an auxiliary graph $G_W$ for each $W \in {\overline \calK_1} \cap {\overline \calK_2}$. This completes the proof.
\end{proof}

The bound on the number of layers turns out to be necessary: the following theorem shows that finding a $D$-diverse committee is intractable even with $3$-layers.

\begin{theorem} \label{thm:inverval-NPh-3layers}
\DCF{} is $\np$-hard even if $D$ is a diversity specification of interval constraints and $\lambda$ is 3-layered.
\end{theorem}
\begin{proof}
We reduce from {\sc 3-Dimensional Matching} ($3$-DM). Given three disjoint sets $X,Y,Z$ of size $n$ and a set $T\subseteq X \times Y \times Z$ of ordered triplets, $3$-DM asks whether there is a set of $n$ triplets in $T$ such that each element is contained in exactly one triplet.

Given an instance $((X,Y,Z),T)$ of $3$-DM, we create one candidate $t_i=(x_i,y_i,B_i)$ for each $t_i \in T$. The set of labels is given by $L=X\cup Y\cup Z$. Each candidate $t_i$ has exactly three labels $\lambda(t_i)=\{x_i,y_i,B_i\}$. The lower bound $b_1(\ell)$ and the upper bound $b_2(\ell)$ of each label $\ell \in L$ are set to be $1$. Lastly, we set $k=n$. It can be easily verified that $W \subseteq T$ is a desired solution for $3$-DM if and only if $W$ is a $D$-diverse committee of size $k$, namely, $|W|=k$, and 
\begin{itemize}
\item $|C_x \cap W| = 1$ for each $x \in X$, 
\item $|C_y \cap W| = 1$ for each $y \in Y$, and 
\item $|C_z \cap W| = 1$ for each $z \in Z$.
\end{itemize}
\end{proof}

Nevertheless, if the number of labels is small (i.e., is taken as the
parameter from the point of view of parametrized complexity theory)
we can compute optimal diverse committees efficiently. The next
theorem expresses this formally (note that interval diversity
specifications can be phrased as linear programs, but this language
allows also some more involved constraints, such as, ``at the research
meeting the number of \emph{senior} researchers should be larger than
the number of \emph{junior} ones, but without taking the \emph{PhD
  students} into account'').

\begin{theorem}\label{thm:DCWD-separable-paraL}
  Let $f$ be separable objective function and let $D$ be a diversity
  specification which can be expressed through a linear program
  $\texttt{LP}$ with the set of variables \mbox{$\{x_{\ell}\colon \ell
    \in L\}$} such that $d \in D$ if and only if $\texttt{LP}$
  instantiated with variables $x_\ell$ giving the numbers of committee
  members with labels $\ell$
% \mbox{$\{x_i = d(i)\colon  i \in L\}$} 
  is feasible.  Then, \DCWD{} is in $\fpt$ with respect to $|L|$.
\end{theorem}

\begin{proof}
To prove the theorem we will use the classic result of \citet{Len83} which states that the problem of solving a mixed integer program is in $\fpt$ for the parameter being the number of integer variables. We create the following mixed integer linear program. For each combination of labels $Y \in 2^L$ we create an integer variable $z_Y$. Each $z_Y$ denotes the number of committee members which have exactly the labels from $Y$ in an optimal committee. Further, for each $Y \in 2^L$ we construct a function $g_{Y}\colon \naturals \to \reals$ in the following way: $g_{Y}(x) = \max_{S\colon |S| = x} f(S)$. In words, $g_{Y}(x)$ gives the value of the best, according to the objective function $f$ and ignoring the distribution constraints, $x$-element committee which consists only of candidates who have sets of labels equal to $Y$. Clearly, since $f$ is separable, the functions $g_{Y}$ are piecewise-linear and concave. We will construct a mixed ILP with the following non-linear objective function:
\begin{align*}
 \text{minimize~} \sum_{Y \in 2^L} g_Y(z_Y) \text{.}  
\end{align*}
We can use piecewise linear concave functions in the minimized objective functions by the result of \citet{mixedIntegerProgrammingFPT} (Theorem 2 in their work)---such programs can still be solved in $\fpt$ time with respect to the number of integer variables. The set of constraints is defined by taking the feasibility linear program $\texttt{LP}$ and for each $\ell \in L$ setting $x_{\ell} = \sum_{Y\colon \ell \in Y}z_Y$. 
\end{proof}

\section{Submodular Objective Functions}
\label{sec:submodular}

The case of submodular objective functions is computationally far more
difficult than that of separable ones. Indeed, even without diversity
constraints computing a winning Chamberlin--Courant committee
(specified through a submodular objective function) is
$\np$-hard~\citep{budgetSocialChoice} and, in general, the best
polynomial-time approximation algorithm for submodular functions is
the classic greedy algorithm~\citep{submodular,fei:j:cover}, which
achieves the $1-\nicefrac{1}{e} \approx 0.63$ approximation
ratio. \footnote{This algorithm starts with an empty committee and
  extends it with candidates one-by-one, always choosing the candidate
  that increases the objective function maximally.}
Adding diversity constraints makes our problems even more difficult.
Nonetheless, we provide a polynomial-time
$\nicefrac{1}{2}$-approximation algorithm for the case of interval
constraints and $1$-laminar labellings.

\begin{theorem}\label{thm:DCWD-1/2-approx}
Let $D$ be a diversity specification of interval constraints. If $\lambda$ is $1$-laminar and $f$ is a monotone submodular function, then \Cref{alg:greedy} gives $\frac{1}{2}$-approximation algorithm for \DCWD{}.
\end{theorem}

\begin{proof}
\citet{Fisher1978} showed that if ${\overline \calK}_D$ is the set of independent sets of a matroid, \Cref{alg:greedy} produces a solution $W$ such that $f(W) \geq \frac{1}{2} f(O^{\prime \prime})$ for any optimal solution $O^{\prime \prime}\in \argmax_{W^{\prime}\in {\overline \calK}_D} f(W^{\prime})$. Now let $O^{\prime}\in \argmax_{W^{\prime}\in {\calK}_D} f(W^{\prime})$. Since $\calK_D \subseteq {\overline \calK}_D$, we have $f(W) \geq \frac{1}{2}f(O^{\prime})$. Clearly, by construction, $W \in \calK(D)$. 
\end{proof}

\paragraph{Balanced Committees}

For the balanced committee model it is possible to achieve notably
stronger results.  Since the balanced case is practically relevant
from practical standpoint, we provide its simpler definition, renaming
it as BCWD.
% In this section we investigate \emph{balanced committees}, where the
% set of candidates is partitioned in two disjoint groups (e.g., men
% and women, or junior and senior researchers) and the committee has
% to contain both groups in equal proportions.
% %. each of size at least equal to $k$. Our goal is to select an optimal committee of size $2k$, with equal number of candidates representing each of the two groups. 
% We define the following variant of the \DCWD{} problem.

% \begin{definition}[\BCWDlong{} (\BCWD{})]
\begin{definition}[\BCWD{}]
Given a set of candidates $C$, 
two subsets $A, B \subseteq C$ such that $A \cap B= \emptyset$ and $A\cup B=C$,
a desired committee size $k = 2k'$, and an objective function $f$,
find a committee $W \subseteq C$ that maximizes $f(W)$ and that satisfies $|W \cap A|=|W \cap B|=k^{\prime}$.
\end{definition}

%Note that for notational convenience, here we search for a committee of size $2k$.

For the case of BCWD, we provide a polynomial-time $1-\nicefrac{1}{e}$
approximation algorithm. Since this is the best possible approximation
ratio for general submodular functions without diversity constraints,
it is also the best one for the balanced setting (formally, the
results without diversity constraints translate because we could
assume that all the candidates with one of the labels have no
influence on the objective value and use the remaining ones to model
an unconstrained submodular optimization problem).
Our algorithm (presented as \Cref{alg:sub}) is very similar
to the classic greedy algorithm, but it considers candidates in pairs.

% We will show the following \Cref{alg:sub} achieves a $(1-1/e)$-approximation ratio for the \BCWD{} problem.

\begin{algorithm}[t]
\SetKwInOut{Input}{input} \SetKwInOut{Output}{output}
\SetKw{And}{and}
\SetKw{None}{None}
\caption{Greedy Algorithm for BCWD}\label{alg:sub}
\Input{$f\colon 2^C \rightarrow \reals$, $A \subseteq C$ and $B \subseteq C$ where $A \cap B= \emptyset$, $|A| \geq k^{\prime}$ and $|B| \geq k^{\prime}$
}
\Output{$W \subseteq C$ where $|W \cap A|=|W \cap B|=k^{\prime}$
}
	\While{$|W|<2k^{\prime}$}{
	choose a pair $e=\{a,b\}$ where $a \in A\setminus W$ and $b \in B\setminus W$ with maximum improvement $f(e|W)$\;
	set $W \leftarrow W \cup e$\;
	}
\end{algorithm}

\begin{theorem}\label{thm:submodular}
Let $f$ be a monotone submodular function. \Cref{alg:sub} gives $(1-\frac{1}{e})$-approximation algorithm for \BCWD{}.
\end{theorem}

\begin{proof}
We first observe the following lemma.
\begin{lemma}\label{lem:subm}
Let $f$ be a submodular function. Then, for any $S \subseteq T$ and any subset $X \subseteq C \setminus T$ of candidates,
\[
f(X|T) \leq f(X|S).
\]
\end{lemma}
\begin{proof}
Let $P=S \cup X$ and $Q=T$. Then $P \cap Q =S$ and $P \cup Q=T \cup X$. Hence, 
\begin{align*}
&f(P \cup Q) + f(P \cap Q) \leq f(P) +f(Q)\\
&\Leftrightarrow f(T \cup X) + f(S) \leq f(S \cup X) +f(T)\\
&\Leftrightarrow f(T \cup X) - f(T) \leq f(S \cup X) -f(S)\\
&\Leftrightarrow f(X|T) \leq f(X|S).
\end{align*}
\end{proof}

Let $W_{i}$ be our greedy solution after $i$ iterations of the algorithm and $e_{i}=\{a_i,b_i\}$ be the pair chosen in the iteration where $a_i \in A$ and $b_i \in B$. Let $O^\prime$ be an optimal solution where $|O^\prime \cap A|=|O^\prime \cap B|=k^{\prime}$. We first prove the following lemma.

\begin{lemma}\label{lem2}
$f(W_{i}) \geq \frac{1}{k^{\prime}} f(O^\prime) + (1-\frac{1}{k^{\prime}})f(W_{i})  $
\end{lemma}
\begin{proof}
Let $O^\prime \setminus W_{i} = \{a^\prime_1,\ldots,a^\prime_p\}\cup \{b^\prime_1,\ldots,b^\prime_q\}$ where each $a^\prime_s \in A$ and $b^\prime_t \in B$ for $s \in [p]$ and $t \in [q]$. Without loss of generality, suppose that $p  \geq q$.
Clearly, we have $p \leq k^{\prime}$. Also, $f(O^\prime) \leq f(O^\prime \cup W_{i})$ by the monotonicity of $f$. 
Now, we let $e^\prime_j=\{a^\prime_j,b^\prime_j\}$ for $j \in [q]$, and observe that
\begin{align*}
&f(O^\prime\cup W_{i})=f(W_{i}) + \sum^q_{j=1}f(e^\prime_j |W_{i} \cup e^\prime_1 \cup  \ldots \cup e^\prime_{j-1}) \\
&+\sum^{p-q}_{j=1}f(a^\prime_{q+j} |W_{i} \cup e^\prime_1 \cup \ldots \cup e^\prime_{q}\cup \{a^\prime_{q+1},\ldots,a^\prime_{q+j-1}\} ) \\
&\leq f(W_{i})  + \sum^q_{j=1}f(e^\prime_j |W_{i}) +\sum^{p-q}_{j=1}f(a^\prime_{q+j} |W_{i})\\
&\leq f(W_{i})  + \sum^q_{j=1}f(e^\prime_j |W_{i})+\sum^{p-q}_{j=1}f(\{a^\prime_{q+j},a_{i+1}\} |W_{i})\\
&\leq f(W_{i})  + qf(e_{i+1} |W_{i})+(p-q)f(e_{i+1}|W_{i}).
\end{align*}
Here the first inequality follows from \Cref{lem:subm}, the second inequality follows from the monotonicity of $f$, and the third inequality follows from the choice of $e_{i+1}$. Now, since $p \leq k^{\prime}$, we have $f(O^\prime\cup W_{i}) \leq f(W_{i}) +k^{\prime} f(e_{i+1} |W_{i})$.
Thus, $f(O^\prime) \leq f(W_{i})  + k^{\prime} f(e_{i+1} |W_{i})$.
\end{proof}

Now by a usual calculation, we have
\begin{align*}
f(W_{k^{\prime}})& \geq \frac{1}{k^{\prime}} f(O^\prime)+ (1-\frac{1}{k^{\prime}}) f(W_{k^{\prime}-1})\\
& \geq \cdots \\
& \geq f(O^\prime)\cdot [1-(1-\frac{1}{k^{\prime}})^{k^{\prime}}]\\
& \geq f(O^\prime) (1-\frac{1}{e}).
\end{align*}
\end{proof}

%\todo{the algorithm can be adapted so as to work with other ratios than half/half.}

%\paragraph{Balanced Committees and Chamberlin--Courant}

We note that \Cref{thm:submodular} is a special case of a much more general result on approximating the Multidimensional Knapsack problem~\citep{KulST13}, which gives the same approximation ratio even for maximizing monotone submodular functions subject to interval constraints consisting only of upper bounds.
Yet, our algorithm is simpler and faster than this general approach.

\Cref{thm:submodular} applies to all submodular
functions. 
However, for some special cases it is possible to achieve
much stronger results. For example, for the Chamberlin--Courant
function we find a polynomial-time approximation scheme (PTAS).

\begin{theorem}\label{thm:chamb-cour}
For each Chamberlin--Courant function there exists a PTAS for \BCWD{}.
\end{theorem}

The main idea behind the proof is to use the PTAS of \citet{SFS15} to compute a committee of size $k'$ and then to
complement it so that it satisfies the diversity constraints. The
specific nature of the algorithm of Skowron et al\@. makes it possible
to do this efficiently and effectively.

\begin{proof}
Let $V$ be the set of $n$ voters used to represent the Chamberlin--Courant objective function $f$. We first make a simple observation that for each $S \subseteq C$ we have $f(S) = \sum_{v_i \in V}\max_{c \in S}(m - \pos_i(c)) \leq \sum_{v_i \in V} m = nm$.
We will use as a black box the result of\citet{SFS15}[Theorem 11] who have shown an algorithm for selecting a committee $S$ of $k$ candidates such that $f(S) \geq nm\left(1 - \frac{2w(k)}{k}\right)$, where $w$ is the Lambert's $W$-function (in particular $w(k) = o(\log(k))$). Now, for each $\epsilon > 0$ we can construct a polynomial-time $(1-\epsilon)$-approximation algorithm for \BCWD{} as follows. We first define a threshold value $k_t \in \naturals$, as a smallest integer such that $\frac{2w(k_t)}{k_t} \leq \epsilon$. Now, if $k < k_t$, then we run a brute-force algorithm trying each $k$-element subset of the set of candidates and, this way, we can find an exactly optimal committee. If $k > k_t$, we run the algorithm of Skowron~et~al. for the instance of our problem without constraints and for the size of the committee equal to $k$; here we have $\epsilon \geq \frac{2w(k_t)}{k_t} \geq \frac{2w(k)}{k}$. Next, we complement the committee returned by this algorithm with some $k$ candidates selected in an arbitrary way so that the diversity constraints are satisfied. Since, adding the candidates can only increase the value of the Chamberlin--Courant objective function, the value of the objective function for such constructed committee is at least equal to $nm(1 - \frac{2w(k)}{k}) \geq nm(1 - \epsilon)$. Thus, this is a $(1-\epsilon)$-approximation committee. 
\end{proof}

$\newline$
\Cref{thm:chamb-cour} also extends to the case of the constant number of labels $\ell_1, \ell_1', \ell_2, \ell_2', \ldots, \ell_p, \ell_p'$ which satisfy the following two conditions:
\begin{inparaenum}[(i)]
\item all the constraints have the following form: for $i \in [p]$ we require the same number of candidates with label $\ell_i$ as those with label $\ell_i'$,
\item for each combination of labels $(r_1, r_2, \ldots, r_p)$ with $r_i \in \{\ell_i, \ell_i'\}$ for each $i \in [p]$, there exist at least $k$ candidates having all labels $r_1, \ldots, r_p$.
\end{inparaenum}

%  It is possible to modify the algorithm from
%  \Cref{thm:chamb-cour} to work for any requirements where for
%  a set of labels $L = \{\ell_1, \ldots, \ell_p\}$ (where $p$ is a
%  fixed constant), we demand that for each $\ell_i \in L$ at least a
%  certain proportion of the committee members have label $\ell_i$
%  (provided that the sum of the proportions does not exceed 100\%).

%\section{Additional Discussion}\label{sec:discussion}

%We conclude our analysis by discussing the following two
%issues. First, we ask how difficult it is to recognize a given
%labeling structure if it is not provided with the problem. While in
%most cases it is natural to assume that the structure would be
%provided (as it would be a common knowledge of the society for which
%we would want to compute the committee), it is interesting to be able
%to derive it automatically.  The second issue is more philosophical in
%spirit and asks how meeting the diversity constraints affects the
%objective value.

\section{Recognizing Structure of the Labels}\label{sec:structure-recognition}
In this section we ask how difficult it is to recognize a given
labeling structure if it is not provided with the problem. While in
most cases it is natural to assume that the structure would be
provided (as it would be a common knowledge of the society for which
we would want to compute the committee), it is interesting to be able
to derive it automatically. 

In the previous sections we have seen that there usually are
polynomial-time algorithms for computing winning committees for
$1$-laminar labellings and, sometimes, there are such algorithms for
$2$-laminar ones. However, $3$-laminar labellings always lead to
$\np$-hardness results.  The same holds for the label-structure
recognition problem. There are algorithms that decide if given
labellings are $1$- or $2$-laminar, but recognizing $3$-layered ones
is $\np$-hard.
In the labeling-recognition problem we are given a set of candidates
$C$, a set of labels $L$, and a labeling function $\lambda$. Our goal
is to recognize if $\lambda$ is $t$-laminar (or $t$-layered), for a
given $t$.

% In \Cref{sec:separable} we have shown that the computational complexity of the \DCWD{} problem depends on the structure of the labels. Thus, it is natural to ask if we can recognize whether a labeling function $\lambda$ is $t$-laminar for a given constant $t$. We will show that this problem is solvable in polynomial time for $t \in \{1, 2\}$, but $\np$-hard for $t = 3$.

\begin{prop}\label{prop:rec-2laminar}
For $t \in \{1, 2\}$ there exists a polynomial-time algorithm for deciding if a given labeling $\lambda$ is $t$-laminar.
The problem of deciding if a given labeling $\lambda$ is $3$-layered is $\np$-hard.
\end{prop}

\begin{proof}
For $t=1$ we can simply check if for each two labels $\ell_1, \ell_2 \in L$ it holds that $C_{\ell_1} \setminus C_{\ell_2} = \emptyset$ or $C_{\ell_2} \setminus C_{\ell_1} = \emptyset$. For $t = 2$ we reduce the problem to 2-satisfiability (2SAT), which can be solved in polynomial time. Let us recall that a labeling $\lambda$ is $2$-laminar if $L$ can be represented as a disjoint union of $L_1$ and $L_2$ such that $\lambda|_{L_i}$ is $1$-laminar for each $i=1,2$.

For each label $\ell \in L$ we create one Boolean variable $x_{\ell}$. Intuitively, if $x_{\ell}= \mathrm{True}$ then $\ell \in L_1$, and if $x_{\ell}= \mathrm{False}$, then $\ell \in L_2$.
For each two labels $\ell_1, \ell_2 \in L$ such that $C_{\ell_1} \setminus C_{\ell_2} \neq \emptyset$ and $C_{\ell_2} \setminus C_{\ell_1} \neq \emptyset$ we include two clauses: $(x_{\ell_1} \vee x_{\ell_2})$ and  $(\neg x_{\ell_1} \vee \neg x_{\ell_2})$---these clauses ensure that $x_{\ell_1} \neq x_{\ell_2}$.

It is apparent that such constructed instance of 2SAT is satisfiable if and only if $\lambda$ is $2$-laminar.

For $t=3$ we give a reduction from the partition into 3 cliques problem, which is $\np$-hard \citep{gar-joh:b:int}. In this problem we are given a graph $G = (V, E)$ and we ask if it is possible to partition the set of vertices $V$ into three sets such that the graphs induced by them are cliques.

For each vertex $x \in V$ we introduce a vertex label $x$, and for each non-adjacent pair of vertices $\{x, y\} \notin E$ we introduce a candidate $c_{\{x, y\}}$ with labels corresponding to $x$ and $y$. First, we will show that if it is possible to partition the so-constructed graph into three cliques, $V_1$, $V_2$ and $V_3$, then the labeling is 3-layered with the layers corresponding to $V_1$, $V_2$ and $V_3$. For that we need to show that the constructed labeling is 1-layered when restricted to $V_i$ for $i \in [3]$. Towards a contradiction, assume this is not the case, i.e., that there exist two labels $x, y \in V_i$ such that $C_x \cap C_y \neq \emptyset$. Let $c \in C_x \cap C_y$. Given that we constructed $c$, we infer that $x$ and $y$ were not adjacent in $G$, a contradiction.    

On the other hand, suppose that the labeling is 3-layered, with layers $V_1$, $V_2$, $V_3$. Then $V_i$ for $i\in[3]$ is a clique. Indeed, if there is a pair of labels $x, y \in V_i$ with $\{x, y\} \notin E$, then  there is a candidate $c_{\{x, y\}} \in C_x \cap C_y$ and hence $C_x \cap C_y \neq \emptyset$, contradicting the fact that $V_i$ is $1$-layered. 
%since for each $x, y \in V_i$ we have $C_x \cap C_y = \emptyset$, we know that we have not added a candidate $c_{\{x, y\}}$ with labels $x$ and $y$ to the set of candidates, so $\{x, y\} \in E$. 
This completes the proof.
\end{proof}

\section{The Price of Diversity}

%By imposing diversity requirements, 
Typically committees that maximize $f$ are not $D$-diverse and, so, we
sometimes need to ``pay a price'' for diversity,\footnote{This notion
  of ``paying a price'' should not be taken literally. In many
  settings maintaining diversity leads to obtaining better societal
  outcomes because many aspects of elected committees are not captured
  by the objective function. Here we are concerned with the loss of
  the value of the objective function, but see it purely as a
  technical concept.} which can be expressed as the ratio between $f$
of the overall optimal committee and $f$ of the optimal $D$-diverse
committee.

\begin{definition}
Let $\calK=\{W\subseteq C: |C|=k\}$ and $\calK_D=\{W\in \calK: W\text{ is $D$-diverse}\}$.
The \emph{price of diversity} of $D$ subject to $f$ is defined as
\begin{align*}
\mathit{pod}(f,D) = \frac{\max_{W\in \calK} f(W)}{\max_{W \in \calK_D} f(W)}.
\end{align*}
\end{definition}

In general, the price of diversity can be unbounded, since the
diversity requirement can force a committee $W$ with minimum objective
value.  For specific diversity specifications, however, one can bound
the price of diversity.

\begin{theorem}\label{thm:pricebound}
Given a $\textsc{Balanced-Committee}$ instance and a submodular objective function $f$, the price of diversity is at most $2$.
\end{theorem}
\begin{proof}
Let $O^*=\argmax_{W\in \calK} f(W)$. Without loss of generality assume that $|O^*\cap A|\geq k^{\prime}$.
Choose $W_1 \subseteq O^*\cap A$ such that $|W_1|=k^{\prime}$ and $f(W_1)$ is maximal.
Since $f$ is submodular, we have $f(O^*)\leq f(W_1)+f(O^*\setminus W_1)$ and, by monotonicity, $f(W_1)\geq \frac{1}{2}\cdot f(O^*)$.
We conclude that for any committee $W$ that contains $W_1$ we have $f(W)\geq \frac{1}{2}\cdot f(O^*)$.
Since there are $D$-diverse committees that contain $W_1$, we have $\mathit{pod}(f,D)\leq 2$.
\end{proof}

In general, we can expect a bounded price of diversity only if either
the diversity requirements are not very restrictive (as it is the case
with balanced committees) or if preferences are to some degree aligned
with the constraints.  A more detailed study of the price of diversity
would be a very interesting research direction.

\section{Related Work}

Our work touches upon many concepts and, thus, is related to many
pieces of research. In this section we briefly mention some of the
most relevant ones.

\citet{conf/aaai/LangS16} considered a model of
diversity requirements that closely resembles our interval
constraints. There are two main differences between their work and
ours:
\begin{inparaenum}[(i)]
\item they do not consider objective functions and
\item their input consists of ``ideal points'' instead of intervals for each label; since there might not exist a committee satisfying such ``exact'' constraints, they focus on finding committees minimizing a certain distance to the ideal diversity distributions.   
\end{inparaenum}

If the labels denote party affiliations of the candidates, the
diversity constraints are one-layered and form instances for the
apportionment problem, where seats in the parliament should be
distributed among the parties (see the book of \citet{BaYo82a} for an overview of the apportionment
problems). Bi-apportionment~\citep{balinski:halshs-00585327,CIS-88648}
can be viewed as an extension of the traditional apportionment to the
case when the diversity constraints are two-layered. However, in all these settings there is no
objective functions, and the goal is only to find a committee
satisfying certain label-based constraints.  For this reason our paper
is even closer the work of \citet{brams1990constrained}, who
%designs an approval-based voting rule for a particular scenario with specific constraints
introduces a specific method based on approval voting that takes diversity constraints into account, which are expressed as quotas for each possible
\emph{tuple} of labels; \citet{potthoff90} and~\citet{straszak1993computer} formulated an ILP for
this method.

Optimization of a given objective functions due to constraints is a
classic problem studied extensively in the literature. For a review of
this literature we refer the reader to the book of \citet{Korte2006}. More specifically, \citet{submodularOverview} provide a comprehensive survey
for the case when the optimized function is submodular. For submodular
functions different types of general constraints are considered,
including matroid and knapsack constraints~\citep{ChaVonZen2014}. A
particularly related case is when the constraints are given for the
size of the committee (see e.g., the works of \citet{QiaShiYuTan17} and the references
inside)---interestingly, this case can be represented in our model,
when we assume that there is a single label assigned to each
candidate, and the constraints are given for the number of occurrences
of this label in the elected committee.
Candidates having positive synergies may induce supermodular (instead of submodular) objective functions. We note that constrained maximization of a supermodular function is equivalent to constrained  minimization of a submodular function, known to be NP-hard~\citep{iwata2009submodular}.

Our model is related to the Multidimensional Knapsack problem with submodular objective functions~\citep{Fre04,Sviridenko04,LeeMNS2009,FloMD10,PucRP10,KulST13}, but differs in a few important aspects. The two biggest differences are:
(i) Multidimensional Knapsack has constraints of the form ``no more than value $D$ on dimension $i$'' (dimensions correspond to labels in our work), whereas our constraints can have more structure (specific quantities of a given label, or upper and lower bounds),
(ii) Multidimensional Knapsack has items that can contribute more than a unit weight to a particular dimension, whereas our candidates only have 0/1 contributions.
Thus, our problem is more general regarding the constraint specification, but less general regarding the structure of the weights of items.

The complexity of selecting an optimal committee without constraints
has been studied extensively. For a general overview of this
literature, we point the reader to a chapter by
\citet{FSST-trends}. Perhaps the most attention
was dedicated to the study of the Chamberlin--Courant
rule~\citep{ccElection}. For instance, it is known that this rule
is $\np$-hard to
compute~\citep{pro-ros-zoh:j:proportional-representation}. The problem
of finding a winning Chamberlin--Courant committee under restricted
domains of voters' preferences was further studied by
\citet{fullyProportionalRepr},
\citet{YCE13a}, \citet{ElkLac15a},
\citet{sko-yu-fal-elk:j:sc-cc}, and \citet{PetLac17a}. Parametrized complexity of the problem
was studied by \citet{fullyProportionalRepr} and
its approximability by \citet{budgetSocialChoice},
\citet{SFS15} and \citet{sko-fal:c:maxcover}.

Finally, we note that \citet{celis2017group} very recently and independently introduced a model for diversity constraints (in their paper refered to as \emph{fairness constraints}) that is similar to our model.
Their paper contains algorithmic results, which are also applicable in our setting.

\begin{table}
	\setlength{\tabcolsep}{14pt}
	\centering
	\begin{tabular}{llccc}
		\toprule
   	     & Label     & Interval  & Independent \\
	Rule & Structure & Constraints & Constraints \\
\midrule
separable &1-laminar   &  P &  P \\
&2-laminar  & P &  NP-hard \\
&3-layer  & NP-hard &  NP-hard \\
&few labels  & FPT & W[1]-hard  \\
          	\midrule
%CC & balanced &  PTAS &  ---  \\
%          	\midrule
submodular & 1-laminar & $0.5$-approx. &  ?  \\
			& balanced & $0.63$-approx. & --- \\ 
\midrule
CC & balanced &  PTAS &  ---  \\
		\bottomrule
	\end{tabular}
	\vspace{3pt}
	\caption{The complexity of computing winning committees for rules 
          of a given type, for the case of candidates with particular 
          label structures, and particular diversity specifications.
          The complexity results for the problem of testing if a feasible
          committee exists are the same as those for computing winning committees. ``Balanced'' label structure refers to the problem
          of computing balanced committees (thus the case of independent constraints is not defined for this setting).}
	\vspace{-5pt}
	\label{table}
\end{table}

\section{Conclusion}
\label{sec:Conclusion}

We studied the problem of selecting a committee of a given size that,
on the one hand, would be diverse (according to a given diversity
specification) and, on the other hand, would obtain as high an
objective value as possible.
We present our results in \Cref{table}. We find that in general
our problem is computationally hard, but there are many tractable
special cases, especially for separable objective functions (which are
very useful for shortlisting tasks, where diversity constraints are
particularly relevant) and for up to $2$-laminar label structures
(which means that dealing with two sets of independent, hierarchically
arranged labels, is feasible). Our work leads to many open
problems. In particular, we barely scratched the surface regarding
approximation of our problems, or their parametrized
complexity. Experimental studies would be very desirable as well.

\subsubsection*{Acknowledgments}
Robert Bredereck was from September 2016 to September 2017 on postdoctoral leave at the University of Oxford, supported by the DFG fellowship BR 5207/2.
Piotr Faliszewski was supported by AGH grant 11.11.230.337 (statutory research).
Ayumi Igarashi was supported by the Oxford Kobe Scholarship. 
Martin Lackner was supported by the European Research Council (ERC) under grant number 639945 (ACCORD) and by the Austrian Science Foundation FWF, grant P25518 and Y698.
Piotr Skowron was supported by a Humboldt Research Fellowship for Postdoctoral Researchers.

%------- Reference ------------------------
\bibliographystyle{aaai}

\end{document}